\documentclass[runningheads,a4paper]{llncs}

\usepackage{amssymb}
\usepackage{graphicx}

\usepackage{url}
\urldef{\mailsa}\path| mahyunst@yahoo.com|
\newcommand{\keywords}[1]{\par\addvspace\baselineskip
\noindent\keywordname\enspace\ignorespaces#1}

\begin{document}

\mainmatter  

\title{Kolmogorov Complexity: Clustering Objects and Similarity}

\titlerunning{Kolmogorov Complexity: Clustering Objects and Similarity}
 
\author{Mahyuddin K. M. Nasution}

\authorrunning{M. K. M. Nasution}
\institute{Mathematic Department, Fakultas Matematika dan Ilmu Pengetahuan Alam\\
Universitas Sumatera Utara, Padang Bulan 20150 USU Medan Indonesia\\
\mailsa\\}

\toctitle{Bulletin of Mathematics}
\tocauthor{}
\maketitle

\begin{abstract}
The clustering objects has become one of themes in many studies, and do not few researchers use the similarity to cluster the instances automatically. However, few research consider using Kommogorov Complexity to get information about objects from documents, such as Web pages, where the rich information from an approach proved to be difficult to. In this paper, we proposed a similarity measure from Kolmogorov Complexity, and we demonstrate the possibility of exploiting features from Web based on hit counts for objects of Indonesia Intellectual.
\end{abstract}

\section{Introduction}
In mathematics, the object is an abstract arising in mathematics, generally is known as \emph{mathematical object}. Commonly they include numbers, permutations, partitions, matrices, sets, functions, and relations. In computer science, these objects can be viewed as binary strings, or strings in forms are words, sentences or documents. Thus we will refer to objects and string interchangeably in this paper. Therefore, sometimes some research also will refer to data as objects or objects as data.

A binary string has the length of the shortest program which can output the string on a universal Turing machine and then stop \cite{grunwald2003}. A universal Turing machine is an idealized computing device capable of reading, writing, processing instructions and halting \cite{sipser1996,montana1998}. The concept of Turing machine is widely used in theoretical computer science, as computational model based on mathematics to approach some problems of real-world. One of problems is about word sense, mainly about context. This problem appears in some applications like machine translation and text summarization, where mostly the existing system needs to understand the correct meaning (semantics relation) and function of words in natural language. This means that the aquasition of knowledge needs a model to abstracts an incomplete information. Therefore, this paper is to address a tool of measurement based on Kolmogorov complexity for finding relations among objects. We first review, in Section 2, the basic terminologies and the concepts. We state, in Section 3, the fundamental results and we discussion property of similarity in Lemma and Theorem. In Section 4, we study a set of objects from Indonesia intellectuals. 

\section{Related Work}
In mathematics, it is more important that objects be definable in some uniform way, for example as sets. Regardless of actual practice, in order to lay bare the essence of its paradoxes, which has traditionally accorded the management of paradox higher priority to objects, and it needs the faithful reflection of the details of mathematical practice as a justification for defining objects. Turing showed this problem in his famous work on the halting problem that it is impossible to write a computer program which is able to predict if some other program will halt \cite{leungyancheong1978,nannen2003}. Thus it is impossible to compute the complexity of a binary string. However there have been methods developed to approximate it, and Kolmogorov complexity is of length of the shortest program which can output the string, where objects can be given literally such as the human can be represented in DNA \cite{powell}.
 
Kolmogorov complexity, also known as algorithm entropy, stochastic complexity, descriptive complexity, Kolmogorov-Chaitin complexity and program-size complexity, is used to describe the complexity or degree of randomness of a binary string. It was independently developed by Andrey N. Kolmogorov, Ray Solomonoff and Gregory Chaitin in the late 1960's \cite{xiao2004,nannen2003}. For an introduction and details see the textbook \cite{li1997}.

\begin{definition}
\label{def:kolkomp}
The \emph{Kolmogorov complexity} of a string $x$, denoted as $K(x)$, is the length, in bits, of the shortest computer program of the fixed reference computing systems that produces $x$ as output.
\end{definition}

The choice of computing system changes the value of $K(x)$ by at most an additive fixed constant. Since $K(x) \stackrel{x}{\rightarrow} \infty$, this additive fixed constant is an ignorable quantity if $x$ is large. One way to think about the Kolmogorov complexity $K(x)$ is to view it as the length (bits) of the ultimate compressed version from which $x$ can be recovered by a general decompression program. The associated compression algorithm transform $x_z$ back into $x$ or a string very close to $x$. A loss compression algorithm is one in which the decompression algorithm exactly computes $x$ from $x_z$ and a loss compression algorithm is one which $x$ can be approximated from given $x_z$. Usually, the length $|x_z|<|x|$. Using a better compressor results in $x_b$ with no redundant information, usually $|x_b|<|x_z|$, etc. So, loss compression algorithms are used when there can be no loss of data between compression and decompression. When $K(x)$ is approximation corresponds to an upper-bound of $K(x)$ \cite{romashchenko2002}. Let $C$ be any compression algorithm and let $C(x)$ be the results of compressing $x$ using $C$.

\begin{definition}
\label{def:pendekatan}
The approximate Kolmogorov complexity of $x$, using $C$ as a compression algorithm, denoted $K_C(x)$, is
\[
K_C(x) = \frac{{\rm Length}(C(x))}{{\rm Length}(x)} + q = \frac{|C(x)|}{|x|}+q
\]
where $q$ is the length in bits of the program which implements $C$.
\end{definition}

If $C$ was able to compress $x$ a great deal then $K_C(x)$ is low and thus $x$ has low complexity. Using this approximation, the similarity between two finite objects can be compared \cite{ziv1978,romashchenko2002}.

\begin{definition}
\label{def:informasi}
The information shared between two string $x$ and $y$, denoted $I(x:y)$, is $I(x:y) = K(y)-K(y|x)$, where $K(y|x)$ is Kolmogorov complexity of $y$ relative to $x$, is the length of the shortest program which can output $y$ if $K(x)$ is given as additional input to the program.
\end{definition}
\begin{table}
\caption{Data compression}
\label{tabel:satu}
\begin{center}
\begin{tabular}{lllccc}\hline
$w$  & Key    & $C(w)$   & $~~|C(w)|~~$ & $~~|w|~~$ & $~~K_P(w)$\cr\hline
$s_1$&$k_1=0100$&$k_1k_2k_1k_3k_1k_4k_5k_6k_5k_5+$&34&40& 0.85\cr
     &$k_2=1101$&$"k_1=0100~k_2=1101~k_3=0001$     &  &  &\cr
     &$k_3=0001$&$k_4=1000~k_5=0101~k_6=1010"$     &  &  &\cr
     &$k_4=1000$&                                 &  &  &\cr
     &$k_5=0101$&                               &  &  &\cr
     &$k_6=1010$&                               &  &  &\cr
$s_2$&$k_1=0100$&$k_1k_1k_1k_1k_1k_7k_1k_8+$&20&32&0.625\cr
     &$k_7=1001$&$"k_1=0100~k_7=1001~k_8=1110"$& &&\cr
     &$k_8=1110$& & &&\cr
$s_3$&$k_5=1001$&$k_5k_6k_5k_5k_1k_1k_1k_7+$&24&32&0.75\cr
     &$k_6=1010$&$"k_5=1001~k_6=1010~k_1=0100$&&&\cr
     &$k_1=0100$&$k_7=1001"$&&&\cr
     &$k_7=1001$&&&&\cr
$s_1|s_2$&$k_2=1101$&$k_1k_2k_1k_3k_1k_4k_5k_6k_5k_5+$&30&40&0.75\cr
         &$k_3=0001$&$"k_2=1101~k_3=0001~k_4=1000$&&&\cr
         &$k_4=1000$&$k_5=0101~k_6=1010"$&&&\cr
         &$k_5=0101$&&&&\cr
         &$k_6=1010$&&&&\cr
$s_1|s_3$&$k_2=1101$&$k_1k_2k_1k_3k_1k_4k_5k_6k_5k_5+$&22&40&0.55\cr
         &$k_3=0001$&$"k_2=1101~k_3=0001~k_4=1000"$&&&\cr
         &$k_4=1000$&&&&\cr\hline
\end{tabular}
\end{center}
\end{table}
Previous classification research using Kolmogorov complexity has been based on the similarity metric developed \cite{vitanyi2005,cilibrasi2005}. Two strings which are similar share patterns and can be compressed more when concatenated than separately. In this way the similarities between data can be measured. This method has been successfully used to classify documents, music, email, and those are of: network traffic, detecting plagiarism, computing similarities between genomes and tracking the evaluation of chain letters \cite{bennett2003,burges1998,cilibrasi2004,cimiano2004,muir2003,patch2003}.

\section{Distance, Metric and Similarity}
Suppose there is a pattern matching algorithm based on compressing each consecutive set of four binary digits (hexadecimal). Let $C$ is the program that performs this compression. For each string $w$, $C$ generates a key of single characters which corresponding to sets of four digits. Let $s_1 = "b_0b_1b_1b_0b_1b_1b_1b_0"$ will generate keys $k_1 = b_0b_1b_1b_0$ and $k_2 = b_1b_1b_1b_0$. The compressed string is composed of the representation plus the key, i.e. $k_1k_2 + "k_1 = b_0b_1b_0b_4~ k_2 = b_1b_1b_1b_0"$. Suppose a second string $s_2 = b_0b_1b_1b_0b_1b_1b_0b_0$ and keys are $k_1 = b_0b_1b_1b_0$ and $k_3 = b_1b_1b_0b_0$, and then the compressed string of $s_2$ is $k_1k_3 + "k_1 = b_0b_1b_0b_4~k_3=b_1b_1b_0b_0$. We can write $C(s_1|s_2) = k_1k_2 + "k_2 = b_1b_1b_1b_0"$. Thus $|C(s_1|s_2)| < |C(s_1)|$ because there is a similar pattern in $s_1$ and $s_2$. For example, we have three strings
\[
\begin{array}{rcl}
s_1 &=& 0100~1101~0100~0001~0100~1000~0101~1010~0101~0101, \cr
s_2 &=& 0100~0100~0100~0100~0100~1001~0100~1110,~ {\rm and} \cr
s_3 &=& 1001~1010~1001~1001~0100~0100~0100~1001. \cr
\end{array}
\]
We can compress each string individually and also the results of compressing $s_1$ using the keys already developed for $s_2$ and $s_3$, Table 1.\par

\[
\begin{array}{rcccccl}
I_C(s_2:s_1) &=& K_P(s_1) - K_P(s_1|s_2) &=& 0.85-0.75 &=& 0.10\cr
I_C(s_3:s_1) &=& K_P(s_1) - K_P(s_1|s_3) &=& 0.85-0.55 &=& 0.30\cr
\end{array}
\]
Thus $I_C(s_3:s_1)>I_C(s_2:s_1)$ is that $s_1$ and $s_3$ share more information than $s_1$ and $s_2$. This defines that the information shared between two strings can be approximated by using a compression algorithm $C$. Therefore, the length of the shortest binary program in the reference universal computing system such that the program computes output $y$ from input $x$, and also ouput $x$ from input $y$, called \emph{information distance} \cite{benett1998,vitanyi2005,cilibrasi2005}.

\begin{definition}
\label{def:jarakinfo}
Let $X$ be a set. A function $E : X \times X \rightarrow {\bf R}$ is called \emph{information distance} (or dissimilarity) on $X$, denoted $E(x,y)$, i.e. $E(x,y) = K(x|y) - \min\{K(x),K(y)\}$ for all $x,y \in X$, it holds:
\begin{enumerate}
\item $E(x,y) \ge 0$,  (non-negativity);
\item $E(x,y) = E(y,x)$, (symmmetry) and;
\item $E(x,y) \le E(x,z)+E(z,y)$, (transitivity).
\end{enumerate}
\end{definition}

This distance $E(x,y)$ is actually a metric, but on properties of information distance these distances that are nonnegative and symmetric, i.e. for considering a large class of admissible distances, whereas computable in the sense that for every such distance $J$ there is a prefix program that has binary length equal to the distance $D(x,y)$ between $x$ and $y$. This means that
\[
E(x,y) \le D(x,y) + c_D
\]
where $c_D$ is a constant that depends only on $D$ but not on $x$ and $y$. Therefore, there are some distances related to one another with features that because it is not suitable. Thus we need to normalize the information distance.

\begin{definition}
\label{def:normaljarak}
Normalized information distance, denoted $N(x|y)$, is 
\[
N(x|y) = \frac{K(x|y)-\min\{K(x),K(y)\}}{\max\{K(x),K(y)\}}
\]
such that $N(x|y) \in [0,1]$.
\end{definition}

Analogously, if $C$ is a compressor and we use $C(x)$ to denote the length of the compressed version of a string $x$, we define {\it normalized compression distance}.

\begin{definition}
\label{def:jarakpadat}
Normalized compression distance, denoted $N_c(x|y)$, is
\[
N_c(x|y) = \frac{C(xy) - \min\{C(x),C(y)\}}{\max\{C(x),C(y)\}}
\]
where for convenience the pair $(x|y)$ is replaced by the concatenation $xy$.
\end{definition}

From Table \ref{tabel:satu}, we calculate $N_c(s_1|s_2) = \frac{30 -20}{34} =  0.294118$, whereas $N_c(s_1|s_3) = \frac{22 - 24}{34} = -0.058824$.

The string give a name to object, like "the three-letter genome of 'love'" or "the text of {\it The Da Vinci Code} by Dan Brown", also there are objects that do not have name literally, but acquire their meaning from their contexts in background common knowledge in humankind, like "car" or "green". The objects are classified by word, the words as objects are classified in the sentences where it represented how the society used the objects, and the words and the sentences are classified in documents.

\begin{definition}
\label{def:kata}
$W = \{w_1,\dots,w_v\}$ represents the number of unique words (i.e., vocabulary) and a word as grain of vocabulary indexed by $\{1,\dots,v\}$.
\end{definition}

\begin{definition}
\label{def:dokumen}
A document $d$ is a sequence of $n$ words denoted by ${\bf w} = \{w_i|i=1,\dots,n\}$, where $w_n$ denotes the $n$th word in a document.
\end{definition}

\begin{definition}
\label{def:korpus}
A {\it corpus} is a collection of $m$ documents denoted by ${\bf D} = \{d_j|j=1,\dots,m\}$, where $d_m$ denotes the $m$th document in a corpus.
\end{definition}

In real world, the corpus is divided two kind: annotated corpus and large corpus. 
The last definition is a representation of body of information physically limited by designing capacity for managing documents. Unfortunately, the modelling collection of document as the annotated corpus not only need more times and much cost to construct and then to manage it, but also this modelling eliminate dynamic property from it. Other side, the collection of digital documents on Internet as web have been increased extremely and changed continuously, and to access them generally based on indexes.

Let the set of document indexed by system tool be $\Omega$, where its cardinality is $|\Omega|$. In our example, $\Omega = \{k_1,\dots,k_8\}$, and $|\Omega| = 13$. Let every term $x$ defines {\it singleton event} ${\bf x}\subseteq \Omega$ of documents that contain an occurence of $x$. Let $P : \Omega\rightarrow [0,1]$ be the uniform mass probability function. The probability of event ${\bf x}$ is $P({\bf x}) = |{\bf x}|/|\Omega|$. Similarly, for terms $x$ AND $y$, the {\it doubleton event} ${\bf x} \cap {\bf y} \subseteq \Omega$ is the set of documents that contain both term $x$ and term $y$ (co-occurrence), where their probability together is $P({\bf x}\cap{\bf y}) = |{\bf x}\cap{\bf y}|/|\Omega|$. Then, based on other Boolean operations and rules can be developed their probability of events via above singleton or doubleton. From Table 1, we know  that term $k_1$ has $|k_1| = 3$ in $s_1$, $|k_1| = 6$ in $s_2$ and $|k_3| = 3$ in $s_3$. Probability of event $k_1$ is $P(k_1) = 3/13 = 0.230769$ because term $k_1$ is occurence in three string as document. Probability of event $\{k_1,k_5\}$ is $P(\{k_1,k_5\}) = 2/13 = 0.153846$ from $s_1$ dan $s_3$.\par
It has been known that the strings $x$ where the complexity $C(x)$ represents the length of the compressed version of $x$ using compressor $C$, for a search term $x$, search engine code of length $S(x)$ represents the shortest expected prefix-code word length of the associated search engine event ${\bf x}$. Therefore, we can rewrite the equation on Definition \ref{def:jarakpadat} as
\[
N_S(x,y) = \frac{S(x|y)-\min\{S(x),S(y)\}}{\max\{S(x),S(y)\}},
\]
called \emph{normalized search engine distance}.

 Let a probability mass function over set $\{\{x,y\}: x,y\in {\cal S}\}$ of searching terms by search engine based on probability events, where ${\cal S}$ is universal of singleton term. There are $|{\cal S}|$ singleton terms, and 2-combination of $|{\cal S}|$ doubleton consisting of a pair of non-identical terms, $x \not= y$, $\{x,y\}\subseteq {\cal S}$. Let $z \in {\bf x}\cap{\bf y}$, if ${\bf x} = {\bf x}\cap {\bf x}$ and ${\bf y} = {\bf y}\cap{\bf y}$, then $z \in {\bf x}\cap {\bf x}$ and $z \in {\bf y}\cap{\bf y}$. For $\Psi = \sum_{\{x,y\}\subseteq{\cal S}}|{\bf x}\cap{\bf y}|$, it means that $|\Psi|\ge|\Omega|$, or $|\Psi|\le\alpha|\Omega|$, $\alpha$ is constant of search terms. Consequently, we can define $p(x) = \frac{P({\bf x})|\Omega|}{|\Psi|} = \frac{|{\bf x}|}{|\Psi|}$, and for ${\bf x} = {\bf x}\cap {\bf x}$, we have $p(x) = \frac{P({\bf x})|\Omega|}{|\Psi|} = \frac{P({\bf x}\cap {\bf x})|\Omega|}{|\Psi|} = p(x,x)$ or $p(x,x) = \frac{|{\bf x}\cap{\bf x}|}{|\Psi|}$.

For $P({\bf x}|{\bf y})$ means a conditional probability, so $p(x) = p(x|x)$ and $p(x|y) = P({\bf x}\cap {\bf y})|\Omega|/|\Psi|$. 
Let $\{k_1,k_5\}$ is a set, there are three subsets contain $k_1$ or $k_5$: $\{k_1\}$, $\{k_5\}$, and $\{k_1,k_5\}$. Let we define an analogy, where $S(x)$ and $S(x|y)$ mean $p(x)$ and $p(x|y)$. Based on normalized search engine distance equation, we have
\begin{equation}
\label{pers:normalisasi}
\begin{array}{rcl}
N_S(x,y) &=& \frac{|{\bf x}\cap{\bf y}|/|{\Psi}| - \min (|{\bf x}|/|{\Psi}|,|{\bf y}|/|{\Psi}|)}{\max(|{\bf x}|/|{\Psi}|,|{\bf y}|/|{\Psi}|)}\cr
&=& \frac{|{\bf x}\cap{\bf y}| - \min (|{\bf x}|,|{\bf y}|)}{\max(|{\bf x}|,|{\bf y}|)}\cr
\end{array}
\end{equation}

\begin{definition}
\label{def:similaritas}
Let $X$ be a set. A function $s : X\times X\rightarrow {\bf R}$ is called \emph{similarity} (or {\it proximity}) on $X$ if $s$ is non-negative, symmetric, and if $s(x,y) \leq s(x,x)$, $\forall x,y \in X$, with an equality if and only if $x=y$.
\end{definition}

\begin{lemma}
\label{lemma:ukuran} If $x,y\in X$, $s(x,y) = 0$ is a minimum weakest value between $x$ and $y$ and $s(x,y) = 1$ is a maximum strongest value betweem $x$ and $y$, then a function $s: X\times X\rightarrow [0,1]$, such that $\forall x,y\in X$, $s(x,y) \in [0,1]$.
\end{lemma}
\begin{proof}
Let $|X|$ is a cardinality of $X$, and $|x|$ is a number of $x$ occured in $X$, the ratio between $X$ and $x$ is $0\leq |{\bf x}|/|{\bf X}|\leq 1$, where $|{\bf x}|\leq|{\bf X}|$. 

The $s(x,x)$ means that a number of $x$ is compared with $x$-self, i.e. $|{\bf x}|/|{\bf x}| = 1$, or $\forall x\in X$, $|{\bf X}|/|{\bf X}|=1$. Thus $1 \in [0,1]$ is a closest value of $s(x,x)$ or called a maximum strongest value. 

In other word, let $z\not\in X$, $|{\bf z}| = 0$ means that a number of $z$ do not occur in $X$, and the ratio between $z$ and $X$ is $0$, i.e., $|{\bf z}|/|{\bf X}| = 0$. Thus $0\in [0,1]$ is a unclosest value of $s(x,z)$ or called a minimum weakest value.

The $s(x,y)$ means that a ratio between a number of $x$ occured in $X$ and a number of $y$ occured in $X$, i.e., $|{\bf x}|/|{\bf X}|$ and $|{\bf y}|/|{\bf X}|$, $x,y\in X$. If $|{\bf X}|=|{\bf x}|+|{\bf y}|$, then $|{\bf x}|<|{\bf X}|$ and $|{\bf y}|<|{\bf X}|$, or $(|{\bf x}|/|{\bf X}|)(|{\bf y}|/|{\bf X}|) = |{\bf x}||{\bf y}|/|{\bf X}|^2 \leq 1$ and $|{\bf x}||{\bf y}|/|{\bf X}|^2 \geq 0$. Thus $s(x,y) \in [0,1]$, $\forall x,y\in X$.
\end{proof}

\begin{theorem}
\label{teorema:similaritas}
$\forall x,y\in X$, the similarity of $x$ and $y$ in $X$ is 
\[
s(x,y) = \frac{2|{\bf x}\cap {\bf y}|}{|{\bf x}|+|{\bf y}|}+c
\]
where $c$ is a constant.
\end{theorem}
\begin{proof}
By Definition \ref{def:jarakinfo} and Definition \ref{def:similaritas}, the main transforms is used to obtain a distance (dissimilarity) $d$ from a similarity $s$ are $d = 1-s$, and from (\ref{pers:normalisasi})
we obtain $1-s = \frac{|{\bf x}\cap{\bf y}| - \min (|{\bf x}|,|{\bf y}|)}{\max(|{\bf x}|,|{\bf y}|)}$.

Based on Lemma \ref{lemma:ukuran}, for maximum value of $s$ is 1, we have $
0 = \frac{|{\bf x}\cap{\bf y}| - \min (|{\bf x}|,|{\bf y}|)}{\max(|{\bf x}|,|{\bf y}|)}
$ 
or $|{\bf x}\cap {\bf y}| = \min (|{\bf x}|,|{\bf y}|)$. For minimum value of $s$ is 0, we obtain
\[
1 = \frac{|{\bf x}\cap{\bf y}| - \min (|{\bf x}|,|{\bf y}|)}{\max(|{\bf x}|,|{\bf y}|)}
\]
or
\[
\begin{array}{rcl}
|{\bf x}\cap{\bf y}| &=& \max(|{\bf x}|,|{\bf y}|)+\min (|{\bf x}|,|{\bf y}|)\cr
&=& |{\bf x}|+|{\bf y}|\cr
\end{array} 
\]
or $1 = (|{\bf x}\cap{\bf y}|)/(|{\bf x}|+|{\bf y}|)$.
We know that $|{\bf x}|+|{\bf y}| > |{\bf x}\cap{\bf y}|$, because their ratios are not 1. If $x = y$, then $|{\bf x}\cap{\bf y}| = |{\bf x}| = |{\bf y}|$, its consequence is 
$1 = (2|{\bf x}\cap{\bf y}|)/(|{\bf x}|+|{\bf y}|)$. Therefore, we have
$
s = \frac{2|{\bf x}\cap{\bf y}|}{|{\bf x}|+|{\bf y}|}+1
$, and $c = 1$, or
\[
s = \frac{2|{\bf x}\cap{\bf y}|}{|{\bf x}|+|{\bf y}|}+c.
\]
\end{proof}

For normalization, we define $|{\bf x}| = \log f(x)$ and $2|{\bf x}\cap{\bf y}| = \log (2f(x,y))$, and the similarity on Definition \ref{def:simmetrikm} satisfies Theorem \ref{teorema:similaritas}.

\begin{definition}
\label{def:simmetrikm}
Let similarity metric I is a function $s(x,y) : X \times X \rightarrow [0,1]$, $x,y\in X$. We define similarity metric M as follow:
\[
s(x,y) = \frac{\log(2f(x,y))}{\log(f(x)+f(y))} 
\]
\end{definition}

In \cite{cilibrasi2005}, they developed Google similarity distance for Google search engine results based on Kolmogorov complexity:
\[
NGD(x,y) = \frac{\max\{\log f(x),\log f(y)\}-\log f(x,y)}{\log N - \min\{\log f(x),\log f(y)\}}
\]
For example, at the time, a Google search for "horse", returned 46,700,000 hits, for "rider" was returned 12,200,000 hits, and searching for the pages where both "rider" and "rider" occur gave 2,630,000. Google indexed $N = 8,058,044,651$ web pages, and $NGD(horse,rider)$ $\approx 0.443$. Using equation in Defenition 10, we have $(s,y) \approx 0.865$, about two times the results of Google similarity distance. At the time of doing the experment, we have 150,000,000 and 57,000,000 for "horse" and "rider" from Google, respectively. While the number of hits for the search both terms "horse" AND "rider" is 12,400,000, but we will not have $N$ exactly, aside from predicting it. We use similarity metric M for comparing returned results of Google and Yahoo!, Table \ref{tabel:dua}.

\begin{table}
\caption{Similarity for two results.}
\label{tabel:dua}
\begin{center}
\begin{tabular}{|l|rrr|r|}\hline
Search engine & $x$ (= "horse") & $y$ (= "rider")& $x$ AND $y$ & $s(x,y)$\cr\hline
Google  & 150,000,000 & 57,000,000 & 12,400,000 & 0.889187\cr
Yahoo! & 737,000,000 & 256,000,000 & 52,000,000 & 0.891084\cr\hline
\end{tabular}
\end{center}
\end{table}

\section{Application and Experiment}
Given a set of objects as points, in this case a set of authors of Indonesian Intellectuals from {\it Commissie voor de Volkslectuur} and their works (Table \ref{tabel:balaipustaka}), and a set of authors of Indonesian Intellectuals from New Writer with their works (Tabel \ref{tabel:pujanggabaru}). 

\begin{table}
\caption{Indonesian Intellectual of {\it Commissie voor de Volkslectuur}}
\label{tabel:balaipustaka} 
\begin{center}
\begin{tabular}{cllccc}\hline
id & Name of Indesian Intellectual & Year & Author & Value & Type\cr\hline
a. & Azab dan Sengsara & 1920 &1 & 0.7348 & 7\cr
b. & Binasa kerna Gadis Priangan & 1931 &1&0.6569&6\cr
c. & Cinta dan Hawa Nafsu  & & 1&0.4357&4\cr
d. & Siti Nurbaya & 1922 & 2&0.5706&6\cr
e. & La Hami & 1924& 2&0.3831&4\cr
f. & Anak dan Kemenakan & 1956& 2&0.5461&5\cr
g. & Tanah Air & 1922 & 3&0.6758&7\cr
h. & Indonesia, Tumpah Darahku & 1928 &3&0.5183&5\cr
i. & Kalau Dewi Tara Sudah Berkata  & & 3&0.4582&5\cr
j. & Ken arok dan Ken Dedes & 1934 &3&0.4922&5\cr
k. & Apa Dayaku karena Aku Seorang Perempuan&  1923& 4&0.5374&5\cr
l. & Cinta yang Membawa Maut & 1926&4&0.8189&8\cr
m. & Salah Pilih & 1928&4&0.7476&7\cr
n. & Karena Mentua & 1932&4&0.6110&6\cr
o. & Tuba Dibalas dengan Susu & 1933&4&0.5918&6\cr
p. & Hulubalang Raja & 1934 & 4&0.7759&7\cr
q. & Katak Hendak Menjadi Lembu & 1935 & 4&0.8424&8\cr
r. & Tak Disangka & 1923 & 5&0.4811&5\cr
s. & Sengsara Membawa Nikmat & 1928 & 5&0.6006&6\cr
t. & Tak Membalas Guna & 1932 & 5&0.5139&5\cr
u. & Memutuskan Pertalian & 1932 & 5&0.6150&6\cr
v. & Darah Muda & 1927 & 6&0.3632&4\cr
w. & Asmara Jaya & 1928 & 6&0.3896&4\cr
x. & Pertemuan & 1927 & 7&0.2805&2\cr
y. & Salah Asuhan & 1928 & 8&0.7425&7\cr
z. & Pertemuan Djodoh & 1933 & 8&0.4376&4\cr
aa. & Menebus Dosa & 1932 & 9&0.4531&5\cr
ab. & Si Cebol Rindukan Bulan & 1934 & 9&0.7516&7\cr
ac. & Sampaikan Salamku Kepadanya & 1935 & 9&0.5786&6\cr\hline
\end{tabular}
\end{center}
\end{table}

\begin{table}
\caption{Indonesian Intellectual of New Writer}
\label{tabel:pujanggabaru}
\begin{center}
\begin{tabular}{cllccc}\hline
id & Name of Indoensian Intelectual & Year & Author & Value & Type\cr\hline
A. & Dian Tak Kunjung Padam & 1932 & i & 0.6372 & 6\cr
B. & Tebaran Mega (kumpulan sajak) & 1935 & i & 0.6189& 6\cr
C. & Layar Terkembang & 1936 & i & 0.7494 & 7\cr
D. & Anak Perawan di Sarang Penyamun & 1940 & i&0.6095&6\cr
E. & Di Bawah Lindungan Ka'bah & 1938 & ii & 0.4302&4\cr
F. & Tenggelamnya Kapal van der Wijck & 1939 & ii&0.7245&7\cr
G. & Tuan Direktur & 1950 & ii&0.6506&6\cr
H. & Didalam Lembah Kehidoepan & 1940 & ii&0.3723&4\cr
I. & Belenggu & 1940 & iii& 0.6007&6\cr
J. & Jiwa Berjiwa & & iii&0.4669&5\cr
K. & Gamelan Djiwa (kumpulan sajak) & 1960 & iii&0.6055&6\cr
L. & Djinak-djinak Merpati (sandiwara) & 1950 & iii&0.6378&6\cr
M. & Kisah Antara Manusia (kumpulan cerpen) & 1953 & iii&0.5380&5\cr
O. & Pancaran Cinta & 1926 & iv&0.5393&5\cr
P. & Puspa Mega & 1927 & iv&0.5681&6\cr
Q. & Madah Kelana & 1931 & iv&0.6477&6\cr
R. & Sandhyakala Ning Majapahit & 1933 & iv&0.6035&6\cr
S. & Kertajaya & 1932 & iv&0.4872&5\cr
T. & Nyanyian Sunyi & 1937 & v&0.5249&5\cr
U. & Begawat Gita & 1933 & v&0.3175&2\cr
V. & Setanggi Timur & 1939 & v&0.5058&5\cr
W. & Bebasari: toneel dalam 3 pertundjukan & & vi&0.5918&6\cr
X. & Pertjikan Permenungan & & vi&0.4988&5\cr
Y. & Kalau Tak Untung & 1933 & vii&0.3611&4\cr
Z. & Pengaruh Keadaan & 1937 & vii&0.3655&4\cr
AA. & Ni Rawit Ceti Penjual Orang & 1935 & viii&0.7906&8\cr
AB. & Sukreni Gadis Bali & 1936 & viii&0.7492&7\cr
AC. & I Swasta Setahun di Bedahulu & 1938 & viii&0.7882&8\cr
AD. & Rindoe Dendam & 1934 & ix&0.6034&6\cr
AE. & Kehilangan Mestika & 1935 & x& 0.5132&5\cr
AF. & Karena Kerendahan Boedi & 1941 & xi&0.8084&8\cr
AG. & Pembalasan & & xi&0.4057&4\cr
AH. & Palawija & 1944 & xii&0.3886&4\cr\hline
\end{tabular}
\end{center}
\end{table}

The authors of {\it Commissie voor de Volkslectuur} are a list of 9 person names:\\

\noindent
$\{$(1) Merari Siregar; (2) Marah Roesli; (3) Muhammad Yamin; (4) Nur Sutan Iskandar; (5) Tulis Sutan Sati; (6) Djamaluddin Adinegoro; (7) Abas Soetan Pamoentjak; (8) Abdul Muis; (9) Aman Datuk Madjoindo$\}$. \\

\noindent
While the authors of New Writer are 12 peoples, i.e., \\

\noindent
$\{$(i) Sutan Takdir Alisjahbana; (ii) Hamka; (iii) Armijn Pane; (iv) Sanusi Pane; (v) Tengku Amir Hamzah; (vi) Roestam Effendi; (vii) Sariamin Ismail; (viii) Anak Agung Pandji Tisna; (ix) J. E. Tatengkeng; (x) Fatimah Hasan Delais; (xi) Said Daeng Muntu; (xii) Karim Halim$\}$.\\

\begin{table}
\begin{tabular}{r|ccccccccc|cccccccccccc|cccccccccccccccccccccccccccccccc}
    & & & & & & & & &  & & & & & & & &v& & & & &&&&&&&&&&&&&&&&&&&&&&&&&&&&& \cr
    & & & & & & & & &  & & &i& & & &v&i& & & &x&&&&&&&&&&&&&&&&&&&&&&&&&&&&&\cr
    & & & & & & & & &  & &i&i&i& &v&i&i&i& &x&i&&&&&&&&&&&&&&&&&&&&&&&&&&&a&a&a\cr
    &1&2&3&4&5&6&7&8&9 &i&i&i&v&v&i&i&i&x&x&i&i&a&b&c&d&e&f&g&h&i&j&k&l&m&n&o&p&q&r&s&t&u&v&w&x&y&z&a&b&c\cr\hline
1   & &6&5&6&7&5&5&5&7 &7&5&6&6&5&6&6&6&6&5&6&4&7&6&4&6&4&6&4&4&5&6&5&4&4&6&4&5&5&4&6&5&5&4&4&5&5&6&5&6&6\cr 
2   &6& &5&6&6&6&5&5&6 &6&5&6&6&5&6&5&6&6&5&6&4&5&6&4&6&4&5&5&5&5&6&5&4&5&5&4&5&5&5&5&5&5&5&5&5&5&6&5&5&5\cr 
3   &5&5& &8&5&5&2&8&4 &6&6&5&6&7&5&4&4&5&4&4&7&5&4&6&6&5&5&7&5&5&5&4&6&6&4&4&5&5&6&6&6&5&7&6&7&6&4&5&4&4\cr 
4   &6&6&8& &7&5&4&9&6 &8&7&6&7&8&6&5&5&6&5&5&7&7&5&7&7&5&6&7&6&7&6&5&8&7&6&6&7&8&7&7&8&7&7&7&8&7&5&6&5&6\cr 
5   &7&6&5&7& &6&5&6&7 &6&5&6&6&5&6&7&7&6&5&7&4&6&6&5&5&4&6&5&5&5&6&5&5&5&6&5&6&5&5&6&5&6&5&5&5&5&7&5&6&6\cr 
6   &5&6&5&5&6& &5&5&6 &6&4&5&5&4&6&7&5&5&5&6&4&4&6&2&4&2&5&4&5&5&5&4&4&2&5&4&4&4&4&4&4&4&4&4&4&4&6&4&5&5\cr 
7   &5&5&2&4&5&5& &4&7 &5&2&5&5&2&5&7&6&5&8&8&2&4&8&2&4&2&4&2&4&4&5&4&2&2&5&4&4&4&2&4&2&4&2&2&2&2&7&4&8&4\cr 
8   &5&5&8&9&6&5&4& &5 &7&7&6&7&8&5&5&5&5&4&5&8&7&4&8&7&7&6&8&5&6&5&5&8&7&5&5&6&7&8&7&8&6&8&8&8&7&4&7&4&6\cr
9   &7&6&4&6&7&6&7&5&  &6&4&6&6&4&6&7&7&6&7&8&2&7&5&8&4&5&4&5&4&4&5&6&5&4&2&6&4&5&5&4&5&4&4&4&4&4&8&5&7&6\cr\hline

i   &7&6&6&8&6&6&5&7&6 & &6&8&8&5&6&6&6&6&5&6&4&6&5&5&6&4&6&5&5&5&6&5&5&5&5&5&5&5&5&5&5&6&5&5&6&6&5&5&5&5\cr 
ii  &5&5&6&7&5&4&2&7&4 &6& &5&5&6&4&4&4&4&2&4&7&6&4&7&6&5&6&7&4&5&4&4&7&7&5&4&6&6&6&6&7&5&7&6&8&6&4&6&2&5\cr 
iii &6&6&5&6&6&5&5&6&6 &8&5& &8&5&6&6&6&6&5&6&4&6&6&5&5&4&5&5&5&5&6&5&5&4&5&5&5&5&4&5&5&5&5&5&5&5&6&5&5&5\cr 
iv  &6&6&6&7&6&5&5&7&6 &8&5&8& &5&6&5&6&6&5&6&5&5&5&5&5&4&5&5&5&5&6&5&5&5&5&5&5&5&5&5&5&5&5&5&6&5&5&5&5&5\cr 
v   &5&5&7&8&5&4&2&8&4 &5&6&5&5& &5&4&5&5&4&5&7&5&4&6&5&4&5&6&4&5&5&4&6&6&5&4&6&6&6&5&6&5&6&6&7&6&4&5&4&5\cr 
vi  &6&6&5&6&6&6&5&5&6 &6&4&6&6&5& &6&6&6&5&6&4&5&6&4&5&4&5&4&5&5&6&5&4&4&5&4&5&4&4&5&4&5&4&4&4&4&6&4&5&5\cr 
vii &6&5&4&5&7&7&7&5&7 &6&4&6&5&4&6& &6&6&7&7&4&4&7&4&4&2&5&2&5&5&5&5&4&2&5&5&5&4&4&4&4&5&4&4&4&4&7&4&7&5\cr 
viii&6&6&4&5&7&5&6&5&7 &6&4&6&6&5&6&6& &6&6&7&4&5&7&4&5&4&5&4&4&5&6&5&2&2&6&4&5&5&4&5&4&5&4&4&4&4&7&4&6&5\cr 
ix  &6&6&5&6&6&5&5&5&6 &6&4&6&6&5&6&6&6& &5&6&2&5&6&4&5&4&5&4&5&5&6&5&4&4&5&5&5&5&4&5&4&5&4&4&4&4&6&4&5&5\cr 
x   &5&5&4&5&5&5&8&4&7 &5&2&5&5&4&5&7&6&6& &7&2&4&7&4&4&2&4&2&4&4&5&4&4&2&5&4&4&4&2&4&2&4&2&2&2&4&7&4&7&4\cr 
xi  &6&6&4&5&7&6&8&5&8 &6&4&6&6&5&6&7&7&6&7& &2&5&9&4&5&4&5&2&4&5&6&5&2&2&6&4&5&5&4&5&4&5&4&4&4&4&8&4&7&6\cr 
xii &4&4&7&7&4&4&2&8&2 &4&7&4&5&7&4&4&4&6&2&2& &5&2&6&5&6&4&7&4&4&2&4&6&6&4&4&5&5&6&6&6&5&7&7&7&6&2&5&2&4\cr 
\end{tabular}
\begin{center}
{\bf Fig. 1}: Matrix of relations.
\end{center}
\end{table}

In a space provided with a distance measure, we extract more information from Web using Yahoo! search engines, then we build the associated distance matrix which has entries the pairwise distance between the objects laying on Definition \ref{def:simmetrikm}. We define some type of relations between author and his/her works in 9 categories: (1) unclose (value $<0.11$), (2) weakest ($0.11\le$ value $<0.22$), (3) weaker ($0.22\le$ value $<0.33$), (4) weak ($0.33\le$ value $<0.44$), (5) midle ($0.44\le$ value $<0.56$), (6) strong ($0.56\le$ value $<0.67$), (7) stronger ($0.67 \le$ value $<0.78$), (8) strongest ($0.78\le$ value $<0.89$), and (9) close (value $\ge 0.89$).

Specifically, some of Indonesia intellectuals of {\it Commissie voor de Volkslectuur} and New Writer be well-known because their works, mainly the works from famous authors which are popularity in society, but also there are visible works because its name is familiar (or same name), for example the story of "Begawat Gita" from Tengku Amir Hamzah, or because the given name frequently appear as words in work of other people or web pages, for example the story of "Pertemuan" from Abas Soetan Pamoentjak, see Table \ref{tabel:balaipustaka} and Table \ref{tabel:pujanggabaru}.

Generally, the appearance of strong interactions in web pages among {\it Commissie voor de Volkslectuur} and New Writer. This situation derive from the time the works appear in the same range of years, or adjacent. In other words, we know that New Writer is the opposition idea of {\it Commissie voor de Volkslectuur} \cite{sutherland1968}, so in any discussion about Indonesia intellectuals, the both always contested and discussed together, see  Fig. 1.

\section{Conclusions and Future Work}

The proposed similarity has the potential to be incorporated into enumerating for generating relations between objects. It shows how to uncover underlying strength relations by exploiting hit counts of search engine, but this work do not consider length of queries. Therefore, near future work is to further experiment the proposed similarity and look into the possibility of enhancing the performance of measurements in some cases.

\end{document}